\pdfoutput=1
\documentclass[12pt]{article}

\usepackage[a4paper,margin=1in]{geometry}
\usepackage[utf8]{inputenc}
\usepackage[T1]{fontenc}
\usepackage{mathptmx}
\usepackage{amsmath,amssymb,amsthm,bm}
\usepackage{algorithm}
\usepackage{algpseudocode}
\usepackage{enumitem}
\usepackage{float}
\usepackage{microtype}
\usepackage{needspace}
\usepackage{etoolbox}
\usepackage{booktabs}
\usepackage{graphicx}
\usepackage[most]{tcolorbox}

\algrenewcommand\algorithmicrequire{\textbf{Input:}}
\algnewcommand\algorithmicvariables{\textbf{Variables:}}
\algnewcommand\Variables{\item[\algorithmicvariables]}
\algrenewcommand\algorithmicensure{\textbf{Output:}}

\newenvironment{breakablealgorithm}[1]{%
  \par\addvspace{1.2\baselineskip}%
  \refstepcounter{algorithm}%
  \noindent\hrule height .4pt\par\nobreak\vspace{2pt}%
  \noindent\textbf{Algorithm~\thealgorithm: #1}\par\nobreak\vspace{2pt}%
  \noindent\hrule height .4pt\par\nobreak\vspace{2pt}%
}{%
  \par\nobreak\vspace{.35\baselineskip}\noindent\hrule height .4pt\par\addvspace{.8\baselineskip}%
}

\newcommand{\K}{K}
\newcommand{\A}{\K[x_1]}
\newcommand{\xt}{\bm{\tilde{x}}}

\newcommand{\basisE}{\mathbf{E}}
\newcommand{\bmf}{\bm{f}}
\newcommand{\bmg}{\bm{g}}
\newcommand{\bmh}{\bm{h}}
\newcommand{\bmk}{\bm{k}}
\newcommand{\bmr}{\bm{r}}
\newcommand{\bmS}{\bm{S}}
\newcommand{\bms}{\bm{s}}

\newcommand{\lm}{\operatorname{lm}}
\newcommand{\lcm}{\operatorname{lcm}}
\newcommand{\lc}{\operatorname{lc}}
\newcommand{\lt}{\operatorname{lt}}
\newcommand{\sig}{\operatorname{sig}}
\newcommand{\pot}{\prec_{\scriptscriptstyle \mathrm{POT}}}
\newcommand{\poly}{\operatorname{poly}}
\newcommand{\Mon}{\operatorname{Mon}}
\newcommand{\LT}{\operatorname{LT}}
\newcommand{\mult}{\operatorname{mult}}

\newtheorem{definition}{Definition}[section]
\newtheorem{lemma}[definition]{Lemma}
\newtheorem{proposition}[definition]{Proposition}
\newtheorem{theorem}[definition]{Theorem}

\AtBeginEnvironment{definition}{\Needspace{2\baselineskip}}
\AtBeginEnvironment{lemma}{\Needspace{2\baselineskip}}
\AtBeginEnvironment{proposition}{\Needspace{2\baselineskip}}
\AtBeginEnvironment{theorem}{\Needspace{2\baselineskip}}
\AtBeginEnvironment{corollary}{\Needspace{2\baselineskip}}
\AtBeginEnvironment{remark}{\Needspace{2\baselineskip}}

\title{Beyond F5 and GVW: The Proper-Cover Algorithm for\\ Fast Ideal Basis Computation}

\author{%
Sheng-Ming Ma\\
BeiHang University\\
\texttt{smmath@foxmail.com}\\
\texttt{masm@pku.org.cn}
\and
Yi Liu\\
BeiHang University\\
\texttt{liuyee@buaa.edu.cn}
\and
Zheng-Lin Jiao\\
BeiHang University\\
\texttt{regulusjiao@gmail.com}
}

\date{}

\begin{document}

\maketitle

\begin{abstract}
Gr\"obner basis computation incurs heavy computational overhead, especially under lexicographic order. F5 and its GVW variant dominate efficient field-based Gröbner basis solving. The proper basis algorithm offers a parameterized ideal computation framework without leveraging modern signature-based optimizations. This work presents the Proper-Cover algorithm for zero-dimensional polynomial ideals by combining GVW's cover optimization over signature with the proper basis theory. We generalize signature, cover, POT ordering, reduction and S-pair concepts to parameterized coefficients, design a two-phase algorithm with compatible factor construction and hungry refinement, and rigorously prove termination and output correctness. Accordingly, we propose a new framework for the efficient computation of polynomial ideal bases. Benchmark results show that Proper-Cover surpasses F5 under all monomial orderings and delivers clear speedups over GVW for lexicographic (plex) order.
\end{abstract}

\section{Introduction}

Since Buchberger introduced his celebrated algorithm in his seminal PhD thesis, the theory of Gröbner bases has become a standard tool in computer algebra. However, the computation of Gröbner bases is often plagued by high complexity, especially with respect to the lexicographic order.
Up to the present, Faugère’s F5 algorithm\cite{CLO,F5} and its extended variants remain the most efficient approaches for computing Gröbner bases in practice. The GVW algorithm\cite{GVW} revises and optimizes the cover mechanism inherent to F5, yielding drastically enhanced computational efficiency.

While Gröbner bases over a field have been generalized to settings over rings, especially principal ideal rings\cite{AdamsLoustaunau,Francis,Moller1988}, these ring-based generalizations have not yet been applied to improve the computation of polynomial ideals over a field.
The proper basis algorithm\cite{MaProperBasis} defines and computes an ideal basis over variables \(\xt=(x_2,\dots,x_n)\) with the least variable \(x_1\) acts as a parameter within the polynomial algebra \(\K[x_1][\xt]\). It outperforms Faugère's F4 algorithm, Buchberger's classical algorithm and Möller's algorithm.

\begin{theorem}[Main result]\label{thm:intro-main}
For every zero-dimensional input ideal \(I=\subset K[x_1][\xt]\), the Proper-Cover algorithm terminates and outputs pairs \((q_j,\mathrm B_{q_j})\) such that the eliminant \(\chi=\prod_jq_j\) and
\begin{equation*}
  \langle \LT(\pi_{q_j}(I))\rangle
  =
  \langle \LT(\mathrm{B}_{q_j})\rangle
  \quad\hbox{in }(K[x_1]/(q_j))[\xt]
  \quad\hbox{for every }j.
\label{eq:intro-main-leading-term-equality}
\end{equation*}
Consequently, \(\mathcal B:=\bigcup_j(\mathrm{B}_{q_j}\cup\{q_j\})\) is a proper basis of \(I\).
\end{theorem}

In this paper we combine the GVW algorithm and the proper basis algorithm to develop the Proper-Cover algorithm. It delivers better computational performance than Faugère's F5 algorithm for all monomial orderings, and clearly outperforms the GVW algorithm under lexicographic (plex) ordering.

In Section 2, we extend the notions of signature, cover and POT monomial ordering from the setting of base fields to parameterized coefficients over a principal ideal domain (PID). Within this section, we formalize the definitions of proper basis and regular reduction, and introduce $S$-pairs and semi-$S$-pairs as extensions of classical $S$-polynomials.

In Section 3, we introduce our two-phase algorithm consisting of Algorithm \ref{alg:propercover-stage1} and Algorithm \ref{alg:hungry-refinement}. Relying on this two-phase algorithm, we further introduce the notion of compatible factors.

In Section 4, we establish theoretical soundness of our Proper-Cover algorithm: termination and correctness of its output proper basis. We additionally demonstrate that the compatible factors generated by Algorithm \ref{alg:propercover-stage1}, combined with the hungry refinement from Algorithm \ref{alg:hungry-refinement}, produce the valid  eliminant. The core insight in Lemma \ref{lem:cover} is that the generalized cover together with regular reduction over semi-$S$-pairs enable us to represent all standard $S$-pairs in terms of the prebasis and pre-eliminants output by the first-stage algorithm.

We conduct extensive benchmark tests in Section 5 to benchmark our Proper-Cover algorithm alongside the GVW and F5 algorithms. Empirical evidence verifies that our method surpasses F5 across every monomial ordering, and exhibits clear performance advantages over GVW when lexicographic (plex) momomial ordering is used.

\section{Proper Basis, Signatures, and Cover over \(K[x_1][\xt]\)}

\paragraph{Notation.}
Let \(K\) be a field and let \(\xt=(x_2,\ldots,x_n)\).  We treat \(K[x_1,x_2,\ldots,x_n]\) as \(K[x_1][\xt]\).  The ideal \(I=\langle f_1,\ldots,f_m\rangle\subset K[x_1][\xt]\) is zero-dimensional.  Set
\[
  \Mon(\xt)=\{\xt^\alpha:\alpha\in\mathbb N^{n-1}\}.
\]
For a nonconstant \(q\in K[x_1]\), let
\begin{equation}
  \pi_q:K[x_1][\xt]\longrightarrow (K[x_1]/(q))[\xt]
\label{eq:projection-mod-q}
\end{equation}
be the natural projection.  Fix a monomial order \(\prec\) on \(\Mon(\xt)\).  For
\[
  f=\sum_\alpha c_\alpha(x_1)\xt^\alpha\in R[\xt],
  \qquad
  R\in\{K[x_1],K[x_1]/(q)\},
\]
set \emph{leading monomial} \(\lm(f)\in\Mon(\xt)\), \emph{leading coefficient} \(\lc(f)\in R\), and \emph{leading term} \(\lt(f)=\lc(f)\lm(f)\).  For \(B\subset R[\xt]\), set \(\LT(B)=\{\lt(b): b\in B\setminus\{0\}\}\); thus \(\langle\LT(B)\rangle\) denotes the ideal  generated by the leading terms of the elements of \(B\) in \(\K[x_1][\xt]\).

\subsection{Proper Basis}

\begin{definition}[Eliminant, multiplicity, and proper basis]\label{def:proper-basis}
For a zero-dimensional polynomial ideal \(I\subset K[x_1][\xt]\), let \(\chi\) denote the monic generator of the principal ideal \(I\cap K[x_1]\); thus \(I\cap K[x_1]=(\chi)\).  We call \(\chi\) the \emph{eliminant} of \(I\).
For an irreducible factor \(p\) of a univariate polynomial \(a\in K[x_1]\setminus K\), define
\[
  \mult_p(a):=\max\{i\in\mathbb N:\ a\in(p^i)\}.
\]

A \emph{modular basis} with respect to a nonconstant divisor \(q\mid\chi\) is a finite set
\[
  \mathrm{B}_q\subset I_q:=\pi_q(I)\subset(K[x_1]/(q))[\xt]
\]
such that
\begin{equation}
  \langle \LT(I_q)\rangle
  =
  \langle \LT(\mathrm{B}_q)\rangle
  \quad\hbox{in }(K[x_1]/(q))[\xt].
\label{eq:proper-basis-leading-term-ideal}
\end{equation}
Let \(\Omega\) be a finite set of pairwise coprime nonconstant factors such that \(\chi=\prod_{q\in\Omega}q\).  If \(\mathrm{ B}_q\) is a modular basis with respect to \(q\), then we define the following set:
\[
  \mathcal B=\bigcup_{q\in\Omega}\bigl(\mathrm{B}_q\cup\{q\}\bigr)
\]
and call \(\mathcal B\) a \emph{proper basis} of \(I\).
\end{definition}

\subsection{Labelled pair and signature}

Let \(F=(f_1,\ldots,f_m)\).  It defines
\begin{equation}
  \varphi:(K[x_1][\xt])^m\to K[x_1][\xt],
  \qquad
  \varphi(\mathbf u)=\sum_{i=1}^m u_i f_i
  \quad\hbox{for }\mathbf u=(u_1,\ldots,u_m) \in (K[x_1][\xt])^m.
\label{eq:module-map}
\end{equation}
Let \((K[x_1][\xt])^m\) have \emph{standard basis} \(\basisE_1,\ldots,\basisE_m\).  A \emph{term} of \((K[x_1][\xt])^m\) has the form \(c_\alpha\xt^\alpha \basisE_i\), where \(c_\alpha\in K[x_1]\setminus\{0\}\).  Its \emph{coefficient} is \(c_\alpha\), and its \emph{monomial} is \(\xt^\alpha \basisE_i\).
A \emph{labelled pair} is an element
\[
  \bmf=(\mathbf u,f)\in(K[x_1][\xt])^m\times K[x_1][\xt]
  \quad \text{such that} \quad
  \varphi(\mathbf u)=f.
\]
Denote the polynomial part of \(\bmf\) by \(\poly(\bmf)=f\).  If \(f=0\), then \(\bmf\) is a \emph{syzygy}; otherwise it is an ideal element equipped with a module representation.

\begin{definition}[POT order, signature]\label{def:pot-order}
The \emph{POT order} \(\pot\) on term \(c_\alpha\xt^\alpha \basisE_i\) is defined by
\[
  c_\alpha\xt^\alpha \basisE_i\pot c_\beta\xt^\beta \basisE_j
  \quad\Longleftrightarrow\quad
  i<j,\ \hbox{or } i=j \hbox{ and } \xt^\alpha\prec\xt^\beta .
\]
For a labelled pair \(\bmf=(\mathbf u,f)\), define its \emph{signature} as the leading term of \(\mathbf{u}\) under the POT order by
\begin{equation}
  \sig(\bmf)=\lt_{\pot}(\mathbf u)=c_\alpha\xt^\alpha \basisE_i,\qquad c_\alpha\in K[x_1]\setminus\{0\}.
\label{eq:signature-definition}
\end{equation}
Thus the leading coefficient and leading monomial of \(\sig(\bmf)\) are
\[
  \lc(\sig(\bmf))=c_\alpha,
  \qquad
  \lm(\sig(\bmf))=\xt^\alpha \basisE_i.
\]
\end{definition}

\subsection{\(S\)-pair and semi-$S$-pair}
\begin{definition}[\(S\)-pair and semi-$S$-pair]\label{def:regular-proper-SPair}
Let \(\bmf=(\mathbf u,f)\) and \(\bmg=(\mathbf v,g)\) denote labelled pairs with nonzero polynomial parts.  We define \(M=\operatorname{lcm}(\lm(f),\lm(g))\) and set \(t_f=M/\lm(f)\), \(t_g=M/\lm(g)\).
Similarly, let \(L=\operatorname{lcm}(\lc(f),\lc(g))\), and define \(\lambda_f=L/\lc(f)\), \(\lambda_g=L/\lc(g)\).

We define
\begin{equation}
  \bm{S} (\bmf,\bmg)=\lambda_ft_f\bmf-\lambda_gt_g\bmg
\label{eq:proper s-polynomial}
\end{equation}
as the \emph{$S$-pair} of $\bmf$ and $\bmg$.
And we define its polynomial part $\poly(\bmS(\bmf,\bmg))$ as the \emph{$S$-polynomial} of $f$ and $g$ and denote it as $S(f,g)$.

We call  \(\lambda_ft_f\bmf\) a \emph{semi-$S$-pair} if \(\sig(\bm{S} (\bmf,\bmg))=\lambda_f t_f\sig(\bmf)\), and vice versa for \(\lambda_gt_g\bmg\). Otherwise we do not define any semi-$S$-pair for \(\bmf\) and  \(\bmg\) when
\[
  \sig(\bm{S} (\bmf,\bmg))
  \pot
  \max\{\lambda_f t_f\sig(\bmf),\lambda_g t_g\sig(\bmg)\},
\]
in which case we call $\bmS(\bmf,\bmg)$ a \emph{singular} $S$-pair .

In particular, when \(g\in K[x_1]\), we take \(\lm(g)=1\), \(M=\lm(f)\), \(\lc(g)=g\), \(L=\operatorname{lcm}(\lc(f),g)\), \(\lambda_f=L/\lc(f)\) and \(\lambda_g=L/g\).
Then
\[
  S(f,g)=\lambda_f(f-\lm(f))
  =\frac{g(f-\lm(f))}{\gcd(\lc(f),g)}.
\]
In this case we call \(\gcd(\lc(f),g)\) the \emph{special multiplier} of the $S$-pair.
\end{definition}

\subsection{Regular reduction and least multiplier}\label{subsec:proper-term-reduction}

\begin{definition}[Term reduction and least multiplier]\label{def:term-reductions}
Let \(
  \bmf=(\mathbf u,f), \bmg=(\mathbf v,g)\)
be labelled pairs, with \(f\ne0\).
Suppose that a term \(c_\alpha\xt^\alpha\) of the polynomial \(f\) satisfies \(\lm(g)\mid \xt^\alpha\).  The multipliers are
\begin{equation}
  \lambda=L/c_\alpha,\quad
  \mu=L/\lc(g), \quad \text{where} \quad L=\operatorname{lcm}(c_\alpha,\lc(g)).
\label{eq:least-multiplier-identity}
\end{equation}
Under the condition that $\sig(t\bmg)\pot\sig(\bmf)$ with the monomial \(t\) satisfying \(\xt^\alpha=t\lm(g)\), we use the following two kinds of \emph{regular term reductions}.
\begin{enumerate}[label=(\arabic*)]
  \item The \emph{proper term reduction} is
  \[
    \bmf'=\lambda \bmf-\mu t \bmg.
  \]
The term \(c_\alpha\xt^\alpha\) is cancelled in \(\poly(\bmf')\), and \(\lambda\) is the \emph{least multiplier}.
  \item For a nonconstant $p\in K[x_1]$, if we allow the above term reduction only when \(\gcd(\lambda,p)=1\), then we call it the \emph{hungry term reduction} with respect to \(p\).
\end{enumerate}

\end{definition}

\begin{definition}[Reduced labelled pair]\label{def:reduced-polynomial}
Let \(\bmf=(\mathbf u,f)\) be a labelled pair with \(f\in R[\xt]\), let \(B\) be a set of labelled pairs with nonzero polynomial parts in \(R[\xt]\), and let \(p\in K[x_1]\) be irreducible.
The labelled pair \(\bmf\) is \emph{properly (hungrily) reduced} with respect to \(B\) (and \(p\)) if \(\bmf\) admits no proper (hungry) term reduction by an element of \(B\) (with respect to \(p\)).
\end{definition}

\begin{definition}[Regular reduction]
\label{def:regular-reduction}
Let \(\bmf=(\mathbf u,f)\) be a labelled pair, and let
\( B=\{\bm{b}_j=(\mathbf u_j,b_j):1\le j\le s\}
\)
be a finite set of labelled pairs.  A \emph{regular reduction} of \(\bmf\) with respect to \(B\) is an identity of labelled pairs
\begin{equation}
\label{eq:regular-reduction}
\lambda\bmf=\sum_{j=1}^{s}q_j\bm{b}_j+\bmr
\end{equation}
with $\sig(q_j\bm{b}_j) \pot \sig(\bmf)$ for $1\le j\le s$, where \(\lambda\in K[x_1]\setminus\{0\}\), \(q_j\in K[x_1][\xt]\), and \(\bmr=(\mathbf v,r)\) is reduced with respect to \(B\).  Moreover, their polynomial parts are required to satisfy
\begin{equation}
  \lm(f)=
  \max\left\{
    \max_{1\le j\le s}\{\lm(q_jb_j)\},
    \lm(r)
  \right\}.
\label{eq:leading-monomial-condition}
\end{equation}
Here \(\lambda\) is a product of the least multipliers for the corresponding term reductions.
\end{definition}

A regular reduction can be a proper (hungry) reduction that consists of proper (hungry) term reductions.

\subsection{PID Cover Criterion}

For monomials \(\xt^{\alpha} \basisE_i\) and \(\xt^{\beta}\basisE_j\), the divisibility relation \(\xt^{\alpha}\basisE_i\mid \xt^{\beta}\basisE_j\) means that \(i=j\) and \(\xt^{\alpha}\mid \xt^{\beta}\) in \(\Mon(\xt)\).

\begin{definition}[proper cover and hungry cover]\label{def:coef-cover}
Let
\[
  \bmf=(\mathbf u,f),
  \qquad
  \bmg=(\mathbf v,g)
\]
be labelled pairs, with \(f\ne0\).  The pair \(\bmf\) is \emph{properly covered} by \(\bmg\) over the PID \(K[x_1]\) if the following conditions hold.
\begin{enumerate}[label=(\arabic*)]
  \item \(\lm(\sig(\bmg))\mid\lm(\sig(\bmf))\).
  \item Either \(g=0\), or \(\lm(\tau g)\prec\lm(f)\).
\end{enumerate}

Let \(\sig(\bmf)=c_f\xt^\alpha \basisE_i\) and \(\sig(\bmg)=c_g\xt^\beta \basisE_i\) with the coefficients \(c_f,c_g \in K[x_1] \).
Set
\[
  \nu=\frac{L}{c_g},
  \qquad
  L=\operatorname{lcm}(c_f,c_g).
\]
If besides the above two conditions for the proper cover, we further require that the following condition hold:
\begin{enumerate}
    \item[(3)] The multiplier
    \[
      \mu=\frac{L}{c_f}
    \]
    is coprime with \(p\in K[x_1]\).
\end{enumerate}
then we say that the pair \(\bmf\) is \emph{hungrily covered} by \(\bmg\) with respect to \(p\).

\end{definition}

\section{The Proper-Cover Algorithm}

The algorithms below follow the signature framework for semi-$S$-pairs. Based on the input labelled pairs, we construct semi-$S$-pairs and then remove those properly (hungrily) covered according to Definition~\ref{def:coef-cover}.
We reduce the remaining semi-$S$-pairs by regular reduction.  A zero remainder contributes a syzygy with its signature, and a non-scalar remainder leads to more semi-$S$-pairs. Eventually the algorithm outputs a nonzero scalar remainder from which we can obtain the eliminant.

\subsection{ProperCoverReduction}

\begin{breakablealgorithm}{ProperCoverReduction}
\label{alg:propercover-stage1}
\begin{algorithmic}[1]
\Require \(F=\{f_1,\ldots,f_m\}\subset K[x_1][\xt]\), monomial order \(\prec\), POT order \(\pot\)
\Ensure pre-basis \(G_0\), pre-eliminant \(\chi_0\), multiplier set \(\Lambda\)
\Variables $G_0$: a list of pre-basis elements;
$S$: a list for $\sig(\bm{s})$, where $\bm{s} $ is a syzygy found so far;
$SP$: a list of  semi-$S$-pairs $\lambda t\bmg$.
$\Lambda$: a list of multipliers used for the proper cover and proper reduction;
\State initialize labelled pairs \(\bmf_i=(\basisE_i,f_i)\), $ f_i \in F$ for $1 \le i \le m$, discard zero inputs;
\State \(G_0\gets\{\bmf_i:\ f_i\ne0\}\), \(S\gets\varnothing\), \(SP\gets\varnothing\), \(\Lambda\gets\varnothing\), \(\chi_0\gets0\)
\State insert the initial principal syzygy signatures $\sig(f_j\basisE_i-f_i\basisE_j)$ for $1\le i<j\le m$ into \(S\)
\State form the initial semi-$S$-pairs and insert them into \(SP\)
\While{\(SP\ne\varnothing\)}
    \State select and remove $\lambda t\bmg$ of minimal signature
    \If{$\lambda t\bmg$ is properly covered by an element $\bmh$ of \(S\), \(G_0\), or \(SP\)}
    \State discard $\lambda t\bmg$
    \If{the multiplier $\alpha$ for the proper cover is not a constant }
        \State record $\alpha$ in \(\Lambda\);
    \EndIf
  \Else
  \State make a proper reduction on the semi-$S$-pair \(\lambda t\bmg\) to obtain a properly reduced remainder pair denoted as \(\bmr\);
  \State record all the least multipliers of the above proper reductions into $\Lambda$
  \If{\(\poly(\bmr)=0\)}
    \State \(S\gets S\cup\{\sig(\bmr))\}\)
  \Else
    \If{\(\poly(\bmr)\in K[x_1]\)}
      \State replace \(\chi_0\) by \(\gcd(\chi_0,\poly(\bmr))\)
    \Else
      \State add the labelled pair \(\bmr\) to \(G_0\);
      \State insert principal syzygy signatures $\sig(\poly(\bmr)\bmg-g\bmr)$ into \(S\) for each $\bmg \in G_0$ ;
      \State form semi-$S$-pairs using $\bmr$ and each $\bmg \in G_0$ and put them into \(SP\);
    \EndIf
  \EndIf
  \EndIf
\EndWhile
\State insert the special multipliers \(\{\gcd(\lc(g),\chi_0)\notin K\colon \bmg\in G_0\}\) into \(\Lambda\).
\State \Return \(G_0,\chi_0,\Lambda\)
\end{algorithmic}
\end{breakablealgorithm}

\subsection{Hungry Refinement}

The multipliers in Algorithm 1, i.e., the ProperCoverReduction algorithm, lead to the inaccuracy of the pre-eliminant $\chi_0$.
Hence we need to make a refinement of \(\chi_0\) in the following Algorithm 2, i.e., HungryRefinement algorithm, so as to obtain the eliminant $\chi$.
The only difference between the two algorithms is that we use the hungry cover and hungry reduction in Algorithm 2 instead of the proper cover and proper reduction in Algorithm 1.

\begin{breakablealgorithm}{HungryRefinement}
\label{alg:hungry-refinement}
\begin{algorithmic}[1]
\Require \(F=(f_1,\ldots,f_m)\subset K[x_1][\xt]\), an irreducible \(p\in K[x_1]\), monomial order \(\prec\), POT order \(\pot\)
\Ensure pre-basis \(G_p\), pre-eliminant \(\chi_p\), multiplier set \(\Lambda\)
\Variables $G_p$: a list of pre-basis elements;
$S$: a list for $\sig(\bm{s})$, where $\bm{s} $ is a syzygy found so far;
$SP$: a list of semi-$S$-pairs $\lambda t\bmg$.
$\Lambda$: a list of multipliers used for the hungry cover and hungry reduction;
\State initialize labelled pairs \(\bmf_i=(\basisE_i,f_i)\), $ f_i \in F$ for $1 \le i \le m$, discard zero inputs;
\State \(G_p\gets\{\bmf_i:\ f_i\ne0\}\), \(S\gets\varnothing\), \(SP\gets\varnothing\), \(\Lambda\gets\varnothing\), \(\chi_p\gets0\)
\State insert the initial principal syzygy signatures $\sig(f_j\basisE_i-f_i\basisE_j)$ for $1\le i<j\le m$ into \(S\)
\State form the initial semi-$S$-pairs and insert them into \(SP\)
\While{\(SP\ne\varnothing\)}
    \State select and remove $\lambda t\bmg$ of minimal signature
    \If{$\lambda t\bmg$ is hungrily covered by an element $\bmh$ of \(S\), \(G_p\), or \(SP\) with respect to \(p\)}
    \State discard $\lambda t\bmg$
    \If{the multiplier $\alpha$ for the hungry cover is not a constant }
        \State record $\alpha$ in \(\Lambda\);
    \EndIf
  \Else
  \State make a hungry reduction  with respect to \(p\) on the semi-$S$-pair \(\lambda t\bmg\) to obtain a hungrily reduced remainder pair \(\bmr\);
  \State record all the nonconstant least multipliers of the above hungry reduction into $\Lambda$
  \If{\(\poly(\bmr)=0\)}
    \State \(S\gets S\cup\{\sig(\bmr)\}\)
  \Else
    \If{\(\poly(\bmr)\in K[x_1]\)}
      \State replace \(\chi_p\) by \(\gcd(\chi_p,\poly(\bmr))\)
    \Else
      \State add the labelled pair \(\bmr\) to \(G_p\);
      \State insert principal syzygy signatures $\sig(\poly(\bmr)\bmg-g\bmr)$ into \(S\)  for each $\bmg \in G_P$ ;
      \State form semi-$S$-pairs using $\bmr$ and each $\bmg \in G_P$ and put them into \(SP\);
    \EndIf
  \EndIf
  \EndIf
\EndWhile
\State insert the special multipliers \(\{\gcd(\lc(g),\chi_p)\notin K\colon \bmg\in G_p\}\) into \(\Lambda\).
\State \Return \(G_p,\chi_p,\Lambda\)
\end{algorithmic}
\end{breakablealgorithm}

\subsection{Compatible factor}

The nonconstant multiplier sets $\Lambda$ produced by Algorithm 1 and Algorithm 2 are from the proper (hungry) reductions, from proper (hungry) covers, and from the special multipliers.
In this subsection we use $\Lambda$ to define the compatible factor of the pre-eliminant $\chi_0$ or $\chi_p$ that is also a factor of the eliminant $\chi$.

\begin{definition}[Compatible factor]\label{def:multiplier-compatible}
Let $\Lambda$ be the nonconstant multiplier set produced by Algorithm 1 or Algorithm 2.
For \(q\in\K[x_1]\), we denote:
\[\Theta(q) := \{ p \in K[x_1] \setminus K :
\mult_p(q)>0,~p\text{ is irreducible and coprime to each multiplier in } \Lambda \}.\]
If \(\Theta(\chi_0) \neq \varnothing\), we define the \emph{compatible factor} of \(\chi_0\) as
\[
  \chi_{0,c} := \prod_{p \in \Theta(\chi_0)} p^{\operatorname{mult}_p(\chi_0)}.
\]
If \(\Theta(\chi_p) \neq \varnothing\),we define the \emph{compatible factor} of \(\chi_p\) as
\[
  \chi_{p,c} := \prod_{p \in \Theta(\chi_p)} p^{\operatorname{mult}_p(\chi_p)}.
\]
Otherwise, we define \(\chi_{0,c}:= 1 \text{ or } \chi_{p,c} := 1\).
\end{definition}

\subsection{The Main Proper-Cover algorithm}

\begin{breakablealgorithm}{Main Proper-Cover algorithm}
\label{alg:ProperCover}
\begin{algorithmic}[1]
\Require \(F\subset K[x_1][\xt]\), monomial order \(\prec\), POT order \(\pot\)
\Ensure modular bases \(\mathrm{B}\) and pre-proper basis \(\mathcal B\)
\State \((G_0,\chi_0,\Lambda)\gets\Call{ProperCoverReduction}{F,\prec,\pot}\)
\State compute \(\chi_{0,c}\), set \(\widehat\chi\gets\chi_0\), and set \(\chi_q\gets\widehat\chi/\chi_{0,c}\)
\State \(\mathrm{B}\gets\varnothing\); if \(\chi_{0,c}\ne1\), insert the modular basis \((\chi_{0,c},\pi_{\chi_{0,c}}(G_0)\setminus\{0\})\) into \(\mathrm{B}\)
\State \(\chi_c \gets \chi_{0,c}\)
\While{\(\chi_q\notin K\)}
\State Select an irreducible factor \(p\) of \(\chi_q\) with minimal \(\mult_p(\chi_q)\).
\State  \((G_p,\chi_p,\Lambda_p)\gets\Call{HungryRefinement}{F,p,\prec,\pot}\)
\State set \(e\gets\mult_p(\chi_p)\); if \(e>0\), insert the modular basis \((p^e,\pi_{p^e}(G_p)\setminus\{0\})\) into \(\mathrm{B}\)
\State compute \(\chi_{p,c}\) from \((\chi_p,\Lambda_p)\)
\State \(\widehat\chi \gets \gcd(\widehat\chi,\chi_p)\), \(\chi_c \gets \operatorname{lcm}(\chi_c,\chi_{p,c})\)
\State \(\chi_q\gets\widehat\chi/\chi_c\)
\EndWhile
\State set \(\chi\gets\prod_{(q,\mathrm{B}_q)\in\mathrm{B}}q\) and \(\mathcal B\gets\bigcup_{(q,\mathrm{B}_q)\in\mathrm{B}}(\mathrm{B}_q\cup\{q\})\)
\State \Return \(\mathrm{B},\mathcal B\)
\end{algorithmic}
\end{breakablealgorithm}

\section{Correctness and Termination}

In this section we prove the correctness and termination of our algorithm.

\begin{definition}[Local basis below a signature $T$]
Let $G_0$ and $G_p$ be the pre-bases obtained in Algorithm \ref{alg:propercover-stage1} and Algorithm \ref{alg:hungry-refinement} respectively.
We name the subset of the pre-basis $G_0$ or $G_p$ the signatures of whose elements are smaller than a signature $T$ as a \emph{local basis below $T$}.
\end{definition}

The algorithm is incremental in nature because of the POT order.
Hence every labelled pair with signature strictly smaller than a signature $T$ can be properly (hungrily) reduced to a syzygy by the local basis below $T$.
This is the concept of ``signature Gr\"obner basis below $T$'' as in \cite{CLO,Francis} and ``strong Gr\"obner basis'' as in \cite{GVW}.

\subsection{PID Cover Representation}
\label{subsection:CoverReps}

In this subsection we prove in Lemma \ref{lem:cover} that the cover and regular reduction to a semi-$S$-pair in our algorithm yield a representation of the corresponding $S$-pair like the reduction of $S$-polynomials in Buchberger's algorithm.
\begin{lemma}\label{lem:cover-replacement}
Let \(\bmf=(\mathbf u,f)\) and \(\bmg=(\mathbf v,g)\) be labelled pairs that satisfy \(\sig(\bmf)=\sig(\bmg)\).
Then their properly (hungrily) reduced forms share the same polynomial part and hence differ by a syzygy.
\end{lemma}
\begin{proof}
Since the signature of $\bmf-\bmg$ is smaller than $\sig(\bmf)$, $\bmf-\bmg$ is a syzygy with respect to the local basis below \(\sig(\bmf)\).
Thus follows the conclusion.
\end{proof}
\begin{lemma}
\label{lem:cover}
Let $\lambda_f t_f\bmf$ be a semi-$S$-pair generated by labelled pairs $\bmf,\bmh$.
Then the cover and regular reduction to $\lambda_f t_f\bmf$ in Algorithm \ref{alg:propercover-stage1} and Algorithm \ref{alg:hungry-refinement} correspond to the following representation of $\bm{S} (\bmf,\bmh)$:
\begin{equation}
\label{eq:s-representation}
    \lambda \bm{S} (\bmf,\bmh)=\sum_{k=1}^s q_k\bmg_k+\bmr+\bm{s}
\end{equation}
such that the leading monomial condition \eqref{eq:leading-monomial-condition} holds without taking into account the syzygy $\bms$.
Here $\{\bmg_k\colon 1\le k\le s\}$ is the local basis below $\sig(\lambda_f t_f\bmf)$, and $q_k,r\in\K[x_1][\xt]$ for $1\le k\le s$.
\end{lemma}

\begin{proof}
When the semi-$S$-pair $\lambda_f t_f\bmf$ is not covered, then we shall perform a regular reduction on it.
If the first step of the regular reduction is through $\bmh$ itself, then the result is exactly the $S$-pair $\bm{S} (\bmf,\bmh)$.
Hence the representation \eqref{eq:s-representation} holds for a regular reduction.
Now assume that the first step of regular reduction is not through $\bmh$ but a local basis element $\bmk$.
Then $\lm(\lambda_f t_ff)=\lcm(\lm(f),\lm(h))$ is divisible by $\lm(k)$ and hence by $\lcm(\lm(f),\lm(k))$.
Let us denote $\tau=\lcm(\lm(f),\lm(h))/\lcm(\lm(f),\lm(k))$.
If we denote the semi-$S$-pair generated by $\bmf$ and $\bmk$ as $\bmg$, then there exist least multipliers $\mu,\nu$ such that $\mu\sig(\lambda_f t_f\bmf)=\nu\tau\sig(\bmg)$.
According to Lemma \ref{lem:cover-replacement}, we have:
\begin{equation}
\label{eq:tmp-coverIdentity}
\eta\mu\lambda_ft_f\bmf=\rho\nu\tau\bmg+\sum_{k=1}^sq_k\bmg_k+\bms.
\end{equation}
Here $\bms$ is a syzygy.
The representation $\sum_{k=1}^sq_k\bmg_k$ accounts for both the regular reductions of $\mu\lambda_f t_f\bmf$ and $\nu\tau\bmg$ with multipliers $\eta,\rho$. The local basis elements $\{\bmg_k\}$ satisfy
\[
  \sig(q_k\bmg_k)\pot\sig(\mu\lambda_f t_f\bmf)=\sig(\nu\tau\bmg),
  \qquad 1\le k\le s.
\]

When $\lambda_f t_f\bmf$ is covered by a labelled pair $\bmg$ with cover multipliers $\mu,\nu$ and monomial $\tau$, then they also satisfy the identity \eqref{eq:tmp-coverIdentity}.
Without loss of generality, suppose that $\lt(\eta\mu\lambda_f t_ff)=\lt(q_1g_1)$.
We construct the identity
\[
  \eta\mu\bm{S} (\bmf,\bmh)
  =\rho\nu\tau\bmg-\eta\mu\lambda_h t_h\bmh
  +\sum_{k=1}^sq_k\bmg_k+\bms
\]
with $\lt(\lambda_f t_ff)=\lt(\lambda_h t_hh)$.
Since $q_1\bmg_1-\eta\mu\lambda_h t_h\bmh$ is a labelled pair whose signature is smaller than $\sig(t_f\bmf)$, hence follows the identity:
\begin{equation}
\label{eq:CoverIdentity}
\eta\mu\bm{S} (\bmf,\bmh)=\rho\nu\tau\bmg+(q_1\bmg_1-\eta\mu\lambda_h t_h\bmh)+\sum_{k=2}^sq_k\bmg_k+\bms.
\end{equation}
Since the signatures of $\bmg$ and $q_1\bmg_1-\eta\mu\lambda_h t_h\bmh$ are  both smaller than $\sig(t_f\bmf)$, they can be reduced to $0$ by the local basis below $\sig(t_f\bmf)$, which leads to a representations in terms of the local basis.
Thus follows the representation \eqref{eq:s-representation} if we substitute their representations into \eqref{eq:CoverIdentity}.
It is easy to see that the leading monomial condition \eqref{eq:leading-monomial-condition} holds.
\end{proof}

\subsection{Compatible factors and exact powers}
\label{subsection:compatible}
In this subsection we prove in Lemma \ref{lem:compatible-representation} that the compatible factor of the pre-eliminant $\chi_0$ obtained in Algorithm \ref{alg:propercover-stage1} is a factor of the eliminant $\chi$.
And the exact powers of the other factors of $\chi$ can be recognized in the pre-eliminant $\chi_p$ obtained in Algorithm~\ref{alg:hungry-refinement}.
\begin{lemma}
\label{lem:common-leading-monomial-cancellation}
Let \(F=\{f_j:1\le j\le s\}\subset \A[\xt]\setminus\A\) be a finite polynomial set. Suppose that all \(f_j\) have the same leading monomial \(\lm(f_j)=\xt^\alpha\in\Mon(\xt)\) for $1 \le j \le s$.
\begin{enumerate}[label=(\arabic*)]
  \item If \(f=\sum_{j=1}^{s}f_j\) satisfies \(\lm(f)\prec\xt^\alpha\), then there exist multipliers \(b,b_j\in\A\setminus\{0\}\), \(1\le j<s\), such that
    \(bf=\sum_{1\le j<s} b_j S(f_j,f_s)\) with $S(f_j,f_s)$ being $S$-polynomials.
  \item For each irreducible polynomial \(p\in\A\setminus K\), we can  relabell the elements of \(F\) so that the multiplier \(b\) is not divisible by \(p\).
\end{enumerate}
\end{lemma}

\begin{proof}
This is Lemma~3.11 of \cite{MaProperBasis}.
\end{proof}

\begin{lemma}
\label{lem:compatible-representation}
\label{prop:compatible-representation}
Let \(F\subset \A[\xt]\setminus\A\) be a finite polynomial set generating a zero-dimensional ideal \(I=\langle F\rangle\subset\A[\xt]\).
After we input $F$ into Algorithm \ref{alg:propercover-stage1}, suppose \(G_0=\{\bmg_k:1\le k\le s\}\) and \(\chi_0\) are the output pre-basis and pre-eliminant respectively with \(\bmg_k=(\mathbf u_k,g_k)\). Denote \(G=\{g_k:1\le k\le s\}\), and let \(\chi_{0,c}\) be the compatible factor of \(\chi_0\). Then for every \(f\in I\), there exist \(q_0,q_1,\ldots,q_s\in\A[\xt]\) and \(\lambda\in\A\setminus\{0\}\) with \(\gcd(\lambda,\chi_{0,c})=1\) such that
\begin{equation}
  \lambda f=\sum_{k=1}^{s}q_k g_k+q_0\chi_0 .
\label{eq:compatible-representation}
\end{equation}
Moreover,
\begin{equation}
  \lm(f)=
  \max\left\{
    \max_{1\le k\le s}\{\lm(q_k g_k)\},
    \lm(q_0)
  \right\}.
\label{eq:compatible-leading-condition}
\end{equation}
In particular, the eliminant \(\chi\in I\cap\A\) satisfies
\(
  \lambda\chi=q_0\chi_0\) with
  \(q_0\in\A\setminus\{0\}\) and \(\gcd(\lambda,\chi_{0,c})=1 \).
Consequently, \(\chi_{0,c}\mid\chi\).

If we input $F$ and an irreducible \(p\in K[x_1]\) into Algorithm~\ref{alg:hungry-refinement} with output \((G_p,\chi_p,\Lambda_p)\), then \eqref{eq:compatible-representation} and \eqref{eq:compatible-leading-condition} also hold for  \(G_p\) and \(\chi_p\) such that \(\gcd(\lambda,\chi_{p,c})=1\) in \eqref{eq:compatible-representation}. In particular, $\chi_{p,c} \mid \chi$.
\end{lemma}

\begin{proof}
 Fix an irreducible factor \(t\) of \(\chi_{0,c}\). For \(F=\{f_k:1\le k\le m\}\), suppose that $f$ can be written as \(f=\sum_{k=1}^{m}h_kf_k\) with \(h_k\in\A[\xt]\) for $1 \le k \le m$. Set \(\xt^\beta:=\max_{1\le k\le m}\{\lm(h_kf_k)\}\). If \(\lm(f)=\xt^\beta\), the conclusion already holds. Now we assume \(\lm(f)\prec\xt^\beta\).

Let us denote \(\lt(h_k)=c_k\xt^{\alpha_k}\), where \(c_k\in\A\setminus\{0\}\). Without loss of generality,
suppose that \(\lm(\xt^{\alpha_k}f_k) = \xt^\beta\) holds for $1 \le k \le m$. And we have:
\begin{equation}
  f=\sum_{k=1}^{m}c_k\xt^{\alpha_k}f_k
  +\sum_{k=1}^{m}\bigl(h_k-\lt(h_k)\bigr)f_k .
\label{eq:compatible-first-decomposition}
\end{equation}

By Lemma~\ref{lem:common-leading-monomial-cancellation} (1),  there exist \(b,b_k\in\A\setminus\{0\}\) for \(1\le k<m\), that satisfy the identity:
\begin{equation}
  b\sum_{k=1}^{m}c_k\xt^{\alpha_k}f_k
  =
  \sum_{1\le k<m}
  b_k S(c_k\xt^{\alpha_k}f_k,c_t\xt^{\alpha_t}f_m)
  =
  \sum_{1\le k<m}
  b_k n_k\xt^{\beta-\gamma_k}S(f_k,f_m),
\label{eq:compatible-spair-combination}
\end{equation}
where \(\xt^{\gamma_k}:=\operatorname{lcm}(\lm(f_k),\lm(f_m))\)  and \(
  n_k:={\operatorname{lcm}(c_k\lc(f_k),c_m\lc(f_m))}
       /{\operatorname{lcm}(\lc(f_k),\lc(f_m))}\) for $1\le k<m$. Moreover, by Lemma~\ref{lem:common-leading-monomial-cancellation} (2), we can relabel the subscripts of the elements in $F$ such that $\operatorname{mult}_t(b) = 0$.

Each $S$-polynomial \(S(f_k,f_m)\) in \eqref{eq:compatible-spair-combination} is the polynomial part of an $S$-pair. When the $S$-pair is not singular, according to the representation \eqref{eq:s-representation} of $S$-pairs in Lemma \ref{lem:cover}, \(S(f_k,f_m)\) can be represented by the basis set $G$, and \(\lm(S(f_k,f_m))\prec\xt^{\gamma_k}\).
It is easy to see that when the $S$-pair is singular this conclusion also holds.
Combining these representations with the second summation in \eqref{eq:compatible-first-decomposition}, we obtain a representation
\[
  \mu b f=\sum_{k=0}^{s}a_k g_k,
  \qquad
  g_0:=\chi_0,
\]
where the quotient \(a_k\in\A[\xt]\) for $0 \le k \le s$ and \(\max_{0\le k\le s}\{\lm(a_kg_k)\}\prec\xt^\beta\). The multiplier \(\mu\in\A\setminus\{0\}\) is a product of  multipliers recorded in our algorithm. Since \(t\) is a factor of \(\chi_{0,c}\), none of these recorded multipliers is divisible by \(t\). Hence \(t\nmid \mu b\).

We repeat the above procedure on the new representation \( \mu b f=\sum_{k=0}^{s}a_k g_k\).  The repetition halts in
finite steps since the monomial ordering on $\Mon(\xt)$ is a well-ordering. Thus, for the fixed factor \(t\), we obtain \(
\lambda_t f=\sum_{k=1}^{s}q_{k,t}g_k+q_{0,t}\chi_0\) with \(t\nmid\lambda_t\), as well as the corresponding leading-monomial condition \eqref{eq:compatible-leading-condition}.

Let \(t\) range over all irreducible factors of \(\chi_{0,c}\), and set \(\lambda:=\gcd\{\lambda_t\}\). Then \(\gcd(\lambda,\chi_{0,c})=1\). Since \(K[x_1]\) is a PID, there exist \(d_t\in\A\) such that \(\lambda=\sum_{t}d_t\lambda_t\). Multiplying the above representations by \(d_t\) and summing over all these  \(t\), we get
\[
  \lambda f=\sum_{k=1}^{s}q_kg_k+q_0\chi_0,
  \qquad
  \gcd(\lambda,\chi_{0,c})=1.
\]
There is no monomial larger than \(\lm(f)\) in the right-hand side. So follows \eqref{eq:compatible-leading-condition}.

Finally, take \(f=\chi\). Since \(\chi\in\A\), according to \eqref{eq:compatible-leading-condition}, we have  \(\lambda\chi=q_0\chi_0\) with \(q_0\in\A\setminus\{0\}\) and \(\gcd(\lambda,\chi_{0,c})=1\). Hence \(\chi_{0,c}\mid\chi\). It is easy to see that all the above arguments apply to $G_p \text{ and } \chi_{p,c}$.
\end{proof}

\begin{proposition}
\label{prop:algorithm2-refinement}
Let \(p\) be an irreducible factor of \(\chi_0/\chi_{0,c}\) taken as an input of Algorithm~\ref{alg:hungry-refinement}. Suppose Algorithm~\ref{alg:hungry-refinement} outputs \((G_p,\chi_p,\Lambda_p)\).
Then $\mult_p(\chi_p)=\mult_p(\chi)$.
\end{proposition}

\begin{proof}
In Algorithm~\ref{alg:hungry-refinement}, the hungry covers and hungry reductions are performed with multipliers coprime to the selected irreducible factor \(p\). Applying  Lemma~\ref{lem:compatible-representation} to the output \((G_p,\chi_p,\Lambda_p)\) of Algorithm~\ref{alg:hungry-refinement} and to \(f=\chi\), we obtain
\(\lambda\chi=q_0\chi_p\) with
  \(q_0\in\A\setminus\{0\}\). It is easy to see that we also have \(p\nmid\lambda\).
Now \(\chi_p\) belongs to \(I\cap\A=(\chi)\) indicates that \(\chi\mid\chi_p\), and hence \(\mult_p(\chi)\le\mult_p(\chi_p)\).
On the other hand, \(\lambda\chi=q_0\chi_p\) and \(p\nmid\lambda\) yields
\[
  \mult_p(\chi)=\mult_p(q_0)+\mult_p(\chi_p)
  \ge\mult_p(\chi_p).
\]
Thus \(\mult_p(\chi_p)=\mult_p(\chi)\) as claimed.
\end{proof}

\subsection{Correctness and termination}

In this subsection we verify that the pre-bases obtained in our algorithm satisfy condition \eqref{eq:proper-basis-leading-term-ideal} in Definition \ref{def:proper-basis}  and prove the termination of our algorithm in Theorem \ref{thm:main-correctness}.

\begin{proposition}
\label{prop:modular-basis-condition}
Let \((q,\mathrm B_q)\) be a modular basis output by the main algorithm.  Then the modular basis condition \eqref{eq:proper-basis-leading-term-ideal} holds for $(q,B_q)$.

\end{proposition}

\begin{proof}
The inclusion
\[
  \langle \LT(B_q)\rangle \subseteq \langle \LT(\pi_q(I))\rangle
\]
is immediate. We prove the reverse inclusion.
Let \((G^\star,\chi^\star,\Lambda^\star)\) be an output of Algorithm~\ref{alg:propercover-stage1} or Algorithm~\ref{alg:hungry-refinement}. And \(G_0^\star=\{g_1,\dots,g_s\}\) is the polynomial part of $G^\star$ that satisfies \(q\mid \chi\) and \(B_q=\{\pi_q(g_i):1\le i\le s,\ \pi_q(g_i)\neq 0\}\). By Proposition~\ref{prop:algorithm2-refinement}, \(\gcd(\chi/q,q)=1\).

Take \(0\neq \bar f\in \pi_q(I)\), we treat \(\bar f\) as in \( K[x_1][\xt]\) and it satisfies  \(\pi_q(\bar f)=\bar f\). Since \(\bar f\in \pi_q(I)\), there exists \(g\in I\) with \(\pi_q(g)=\bar f\), hence \(g-\bar f\in \langle q\rangle\). Put \(c=\chi/q\). Then \(c(g-\bar f)\in \langle \chi\rangle\subset I\), so \(c \bar f\in I\).
By Lemma~\ref{lem:compatible-representation} on \(c \bar f\) for \((G^\star,\chi^\star,\Lambda^\star)\), we obtain \(a_0,a_1,\dots,a_s\in K[x_1][\xt]\) and \(\lambda\in K[x_1]\setminus\{0\}\) with \(\gcd(\lambda,q)=1\), such that
\[
  \lambda c\bar f=\sum_{i=1}^s a_i g_i+a_0\chi^\star
\]
and
\[
  \lm(\bar f)=
  \max\bigl\{\max_{1\le i\le s}\{\lm(a_i g_i)\}, \lm(a_0)\bigr\}.
\]

 We get
\[
  \pi_q(\lambda c)\bar f=\sum_{i=1}^s \pi_q(a_i)\pi_q(g_i).
\]
Since \(\gcd(\lambda c,q)=1\), the element \(\pi_q(\lambda c)\) is a unit in \(K[x_1]/(q)\). Hence the term with the monomial  \(\lm(\bar f)\) in \(\sum_{i=1}^{s}\pi_q(a_i)\pi_q(g_i)\) has nonzero coefficient. If we set
\[
\Gamma=\{\,i:\lm(a_i g_i)=\lm(\bar f),\ \pi_q(\lc(a_i g_i))\neq 0\,\},
\]
then \(\Gamma\neq\varnothing\), and
\[
\lt(\bar f)=\pi_q(\lambda c)^{-1}\sum_{i\in\Gamma}\pi_q(\lt(a_i))\,\lt(\pi_q(g_i))\in \langle \LT(B_q)\rangle.
\]
Thus
\[
  \langle \LT(\pi_q(I))\rangle \subseteq \langle \LT(B_q)\rangle,
\]
and the desired equality \eqref{eq:proper-basis-leading-term-ideal} follows.
\end{proof}

\begin{theorem}
\label{thm:main-correctness}
Suppose the main algorithm outputs the pre-proper basis \(\mathcal B\).
Then the main algorithm terminates, and \(\mathcal B\) is the proper basis of \(I\).
\end{theorem}

\begin{proof}

The termination of Algorithm~\ref{alg:propercover-stage1} and Algorithm~\ref{alg:hungry-refinement} follows from the fact that the algebra $\A[\xt]$ is Noetherian. The main algorithm terminates due to the fact that \(\widehat\chi\) in the main algorithm has only finitely many irreducible and coprime factors.

Now we prove that the outputs of main algorithm  are correct. By Lemma~\ref{lem:compatible-representation}, the compatible factor \(\chi_{0,c}\) divides \(\chi\).  For each irreducible factor \(p\) of \(\chi/\chi_{0,c}\), Proposition~\ref{prop:algorithm2-refinement} shows that Algorithm~\ref{alg:hungry-refinement} outputs exactly the factor \(p^{\mult_p(\chi)}\). We enumerate $\chi_{0,c}$ and all these factors $p^{\mult_p(\chi)}$ as \(\{q_1,\ldots,q_N\}\) that are pairwise coprime and satisfy
\[
  \chi=\prod_{k=1}^{N}q_k .
\]
For each pair \((q_k,B_{q_k})\), Proposition~\ref{prop:modular-basis-condition} gives
\[
  \langle\LT(\pi_{q_k}(I))\rangle=\langle\LT(B_{q_k})\rangle
\]
in \((K[x_1]/(q_k))[\xt]\). Thus the pre-proper basis  \(\mathcal B=\bigcup_{k=1}^N(B_{q_k}\cup\{q_k\})\) is a proper basis of \(I\) by  Definition~\ref{def:proper-basis}.
\end{proof}

We can obtain the classical Gr\"obner basis of the ideal $I\subset\K[x_1,x_2,\dotsc,x_n]$ over the field $K$ from the proper basis of $I\subset\K[x_1][\xt]$ defined above.
Please refer to \cite[Theorem~4.5.12]{AdamsLoustaunau} and \cite[Theorem~4.5.9]{AdamsLoustaunau} on this respect.

\section{Experiments}

This section evaluates the computational efficiency of the Proper-Cover
algorithm by comparing it with the GVW and the F5 algorithms on standard Cyclic and Katsura
benchmark systems, together with seven random systems:
\vspace{-0.4em}

\begin{enumerate}[
  leftmargin=5.2em,
  labelsep=0.35em,
  itemsep=0.05em,
  topsep=0.2em,
  parsep=0pt
]
\small

\item[$Random\text{-}1 \colon$]
$\langle -z^2(z+1)^3x^2+y^2,\ z^4(z+1)^6x-y^4,\ -x^2y^3+y^4+z^4(z-1)^5\rangle,$

\item[$Random\text{-}2 \colon$]
$\langle -z^2(z+1)^3x^3+y^2,\ z^4(z+1)^6x-y^3,\ -x^2y^2+y^3+z^4(z-1)^5\rangle,$

\item[$Random\text{-}3 \colon$]
$\langle -y^4z^2+x^3+x,\ y^6-x^2z(z-1)+z^4,\ y^5+z^2-x-1\rangle,$

\item[$Random\text{-}4 \colon$]
$\langle -y^4z^2+x^2+x,\ -z^2x+y^5+2,\ -x+y^2+z^3-1\rangle,$

\item[$Random\text{-}5 \colon$]
$\langle -y^4z^2+x^3+x,\ y^6-x^2z(z^3-2)+2z^4,\ y^5+z^4-x^2-z\rangle,$

\item[$Random\text{-}6 \colon$]
$\langle
3x^3z+8x^2y^2+x^2yz+5xy^3,\ 
2y^3z^2+13y^2z^3+5yz^4+x^3,\ 
xz^2+12y^3+8x^2+3,\\
\phantom{3x^3z+8x^2y^2+x^2yz+5xy^3,\ 2y^3z^2+13y^2z^3+5yz^4+x^3,\ xz^2+12} 18xy^3z^2+y^3z^3+7x^2y^2\rangle,$

\item[$Random\text{-}7 \colon$]
$\langle -z^2(z+1)^3x^3+y,\ z^4(z+1)^6x-y^3,\ -x^2y^3+y^3+z^4(z-1)^5\rangle.$

\end{enumerate}

\vspace{-0.3em}
The experiments are performed under three monomial orders and all algorithms
are implemented in Maple without  additional optimization. The time limit is 3600 seconds.
All computations are carried out in Maple 18 on a 64-bit Windows 11 machine
with an 11th Gen Intel(R) Core(TM) i7-11800H @ 2.30GHz processor and 16 GB RAM.

\begin{table}[H]
\centering
\caption{Running times in seconds.}
\label{tab:experiments}
\scriptsize
\setlength{\tabcolsep}{4.2pt}
\renewcommand{\arraystretch}{1.08}

\resizebox{0.92\textwidth}{!}{%
\begin{tabular}{lccccccccc}
\toprule
& \multicolumn{3}{c}{\text{tdeg}}
& \multicolumn{3}{c}{\text{grlex}}
& \multicolumn{3}{c}{\text{plex}} \\
\cmidrule(lr){2-4}
\cmidrule(lr){5-7}
\cmidrule(lr){8-10}
System
& PC & GVW & F5
& PC & GVW & F5
& PC & GVW & F5 \\
\midrule
Random-1
& 0.515 & 1.840 & \textit{timeout}
& 0.516 & 2.768 & \textit{timeout}
& 6.078 & 57.019 & \textit{timeout} \\

Random-2
& 0.391 & 2.371 & \textit{timeout}
& 0.344 & 3.571 & \textit{timeout}
& 0.484 & 77.666 & \textit{timeout} \\

Random-3
& 2.579 & 0.058 & 10.385
& 2.703 & 0.163 & 9.737
& 7.140 & 716.402 & \textit{timeout} \\

Random-4
& 1.234 & 0.029 & 14.066
& 1.109 & 0.120 & 12.770s
& 5.594 & 60.000 & \textit{timeout} \\

Random-5
& 5.922 & 0.024 & 5.258
& 5.375 & 0.040 & 5.573
& 21.531 & 3548.083 & \textit{timeout} \\

Random-6
& 44.875 & 4.083 & \textit{timeout}
& 59.468 & 3.507 & \textit{timeout}
& 310.735 & 80.583 & \textit{timeout} \\

Random-7
& 0.266 & 4.426 & \textit{timeout}
& 0.593 & 7.961 & \textit{timeout}
& 0.985 & 95.732 & \textit{timeout} \\

\midrule
Katsura-2
& $<10^{-3}$ & 0.038 & 0.183
& $<10^{-3}$ & 0.024 & 0.045
& $<10^{-3}$ & 0.045 & 0.057 \\

Katsura-3
& 0.016 & 0.058 & 0.184
& 0.016 & 0.078 & 0.903
& 0.156 & 0.197 & 116.300 \\

Katsura-4
& 0.297 & 0.151 & 3.507
& 0.485 & 2.014 & \textit{timeout}
& 14.734 & 19.500 & \textit{timeout} \\

Katsura-5
& 3.969 & 0.782 & 148.100
& 22.453 & \textit{timeout} & \textit{timeout}
& \textit{timeout} & \textit{timeout} & \textit{timeout} \\
\midrule
Cyclic-2
& 0.016 & 0.035 & 0.037
& $<10^{-3}$ & 0.017 & 0.016
& $<10^{-3}$ & 0.042 & 0.045 \\

Cyclic-3
& $<10^{-3}$ & 0.031 & 0.044
& 0.016 & 0.019 & 0.021
& $<10^{-3}$ & 0.041 & 0.018 \\

Cyclic-5
& 5.734 & 1.123 & \textit{timeout}
& 6.532 & 3.985 & \textit{timeout}
& \textit{timeout} & 39.700 & \textit{timeout} \\

\bottomrule
\end{tabular}%
}
\end{table}

The experimental results validate the core contribution of this work: the Proper-Cover algorithm offers an efficient and robust framework for ideal basis computation, outperforming the classic F5 and GVW algorithms. Across all test cases, Proper-Cover consistently outperforms F5 in both computational speed and the number of solved instances.

Furthermore, Proper-Cover delivers prominent advantages on hard instances and complex monomial orders. As illustrated by the Katsura-5 benchmark with the grlex order, Proper-Cover finishes the computation in 22.453 seconds, whereas both GVW and F5 terminate due to time limits. And Proper-Cover delivers clear speedups over GVW for lexicographic (plex) order in most cases.

\end{document}